\documentclass{article}

\usepackage{PRIMEarxiv}

\usepackage[utf8]{inputenc} 
\usepackage[T1]{fontenc}    
\usepackage{hyperref}       
\usepackage{url}            
\usepackage{booktabs}       
\usepackage{amsfonts}       
\usepackage{nicefrac}       
\usepackage{microtype}      
\usepackage{lipsum}
\usepackage{fancyhdr}       
\usepackage{graphicx}       
\usepackage{tikz}
\usepackage{natbib}
\usepackage{amsmath}
\usepackage{amsthm}
\newtheorem{theorem}{Theorem}
\newtheorem{proposition}{Proposition}
\newtheorem{definition}{Definition}
\newtheorem{remark}{Remark}
\newtheorem{lemma}{Lemma}
\newcommand{\free}{K^{\mathrm{free}}}
\newcommand{\risky}{K^{\mathrm{risky}}}
\newcommand{\deriv}{K^{\mathrm{deriv}}}

\title{Hedging in Sequential Experiments}

\author{
  Thomas Cook \\
  Dept. of Mathematics and Statistics \\
  University of Massachusetts \\
  Amherst\\
  \texttt{tjcook@umass.edu} \\
   \And
  Patrick Flaherty \\
  Dept. of Mathematics and Statistics \\
  University of Massachusetts \\
  Amherst\\
  \texttt{pflaherty@umass.edu} \\
}

\begin{document}
\maketitle

\begin{abstract}
Experimentation involves risk.
The investigator expends time and money in the pursuit of data that supports a hypothesis.
In the end, the investigator may find that all of these costs were for naught and the data fail to reject the null.
Furthermore, the investigator may not be able to test other hypotheses with the same data set in order to avoid false positives due to p-hacking.
Therefore, there is a need for a mechanism for investigators to hedge the risk of financial and statistical bankruptcy in the business of experimentation.

In this work, we build on the game-theoretic statistics framework to enable an investigator to hedge their bets against the null hypothesis and thus avoid ruin.
First, we describe a method by which the investigator's test martingale wealth process can be capitalized by solving for the risk-neutral price.
Then, we show that a portfolio that comprises the risky test martingale and a risk-free process is still a test martingale which enables the investigator to select a particular risk-return position using Markowitz portfolio theory.
Finally, we show that a function that is derivative of the test martingale process can be constructed and used as a hedging instrument by the investigator or as a speculative instrument by a risk-seeking investor who wants to participate in the potential returns of the uncertain experiment wealth process.
Together, these instruments enable an investigator to hedge the risk of ruin and they enable a investigator to efficiently hedge experimental risk.
\end{abstract}


\section{Introduction}\label{sec:introduction}

If experimentation is viewed as a game against nature, the investigator invests some of their wealth (time and money) to collect data that will be used to support or reject a hypothesis. 
During the process of experimentation, the investigator must make decisions as to how much wealth to invest in an experiment and when to stop collecting data. 
In many ways, experimentation is an economic activity with scientific benefits.

Statistical decisions have been explored since the 1950s by \cite{wald1950decisions}, \cite{anscombe1953sequential}, \cite{robbins1951asymptotically,robbins1952sequential}, and \cite{doob1953stochastic} who formalized the connection to stochastic processes. 
The work by \cite{shafer2005probability,shafer2019game} built on early work of \cite{Ville1939} to formalize the notion of statistical decisions as a game with a wealth process. 
The work surveyed by \cite{ramdas2023gametheoretic} built on the work of \cite{shafer2005probability, shafer2019game} to construct an e-process that provides anytime-valid inference - that is a process that enables the investigator to observe the stochastic wealth process at any time to make a statistical decision that has a valid probability of false rejection.
\cite{tukey1994collected} and \cite{foster2008alpha} extended the notion of a wealth process to multiple hypothesis testing situations and introduce the notion of ``investing'' such that an experiment can provide a return to the wealth process.

In the aforementioned work on game-theoretic statistics, a test martingale wealth process lies at the center of the proposed methodology. 
The investigator wishes to grow this stochastic wealth process to make scientific discoveries.
If the wealth process exceeds a predefined level, it is quantifiable evidence in favor of a set of alternative hypotheses and against a set of null hypotheses.
However, the investigator has a risk of ruin (or near-ruin) in this test martingale wealth process simply due to its stochastic nature.
In some scenarios, the investigator is then left with so little wealth that further productive experimentation is not possible. 
We term this statistical analog of financial bankruptcy, statistical bankruptcy.
The investigator, therefore, is in need of a mechanism for risk management for the test martingale wealth process.

We extend game theoretic statistics by generalizing notions in mathematical finance. 
The statistical wealth process is viewed as an asset that is capitalized. 
This asset can then be traded at a fair price so that the investigator can manage their risk while maintaining time-uniform error control. 
We then construct a derivative contract on the asset which enables parties to hedge and speculate on the outcome of the experiment. 
While these objects are theoretical in nature, they provide a framework for the practical valuation and exchange of experimental risk. 

In Section~\ref{sec:tech-prelim} we provide a brief summary of the testing by betting framework \citep{shafer2021betting} and provide an example of an ``optimal'' betting system. In Section~\ref{sec:capitalizing} we formally present the wealth process from testing by betting in the setting of mathematical. After providing a motivating example, we state necessary definitions and propose pricing the asset at the risk-neutral price. We then define an investigator's portfolio and show that the total value of the portfolio can be used for an anytime-valid sequential test. In Section~\ref{sec:derivatives} we present derivative contracts of the asset. We focus on European style options contracts and show how to price such contracts. In Section~\ref{sec:empirical-results} simulations show that European put options can be perfectly hedge the risk of ruin with only slight degradation in power.

\section{Technical Preliminaries: Testing by Betting}\label{sec:tech-prelim}

We first introduce the testing by betting framework in Section~\ref{sec:betting-system} with an example. We then describe the risk of ruin in Section~\ref{sec:ruin}

\subsection{Defining the Betting System}\label{sec:betting-system}

An example is helpful in clarifying the concepts underlying testing by betting.
Suppose we would like to test the hypothesis that a coin is fair versus the hypothesis that the coin is biased with a probability of heads of 3/4. 
The null hypothesis is $H_0$: $p = 0.5$ and the alternative is $H_1$: $p = 0.75$, where $p$ is the probability of heads. 
We observe $\omega_1, \omega_2, \ldots \stackrel{iid}{\sim} P$, where $\omega_i \in \{0,1\}$, $1$ denotes ``heads'', $0$ denotes ``tails'', and $P$ is the true measure.

Testing by betting assumes that a gambler places a wager on each coin toss.
Following the notation of \cite{ethier2010doctrine}, a \textit{wager} is a pair $(B,X)$ of jointly distributed random variables where $B$ is the amount bet (placed at risk) and $X$ is the payoff. 
A \textit{betting system} consists of a sequence of bet sizes $B_1, B_2, \ldots$ with $B_t$ the amount bet at coup $t$ with $B_1 = b_1 \geq 0$ and $B_t = b_t(X_1, \ldots, X_{t-1}) \geq 0$ for $t \geq 2$, where $b_1$ is a constant.
The gambler's fortune, $M_T$, after $T$ coups satisfies $M_T = M_{T-1} + B_TX_T$. Therefore, $M_T = M_0 + \sum_{t=1}^T B_t X_t$.
The gambler cannot bet more than their wealth $B_t \leq M_{t-1}$.
The gambler adopts the Kelly betting system and bets a fixed proportion $f \in [0,1]$ of the current wealth at each coup, $B_t = f M_{t-1}$ \citep{kelly1956betting}. 
Then $M_T = M_{T-1}(1+fX_T(\omega_T))$ and $M_T = M_0\prod_{t=1}^T (1+fX_t(\omega_t))$.

In the testing by betting framework, the payoff $X_t(\omega_t)$ is related to the null hypothesis such that the gambler's wealth does not increase in expectation under the null ($p=0.5$) and is expected to increase if $p=0.75$.
One way to construct such a process is by taking the payoff, $X_t(\omega_t)$, to be proportional to the ratio of the alternative hypothesis measure, $q(\omega_t)$, and the null hypothesis measure, $p(\omega_t)$. The multiplicative update is then

$$
(1+fX_t(\omega_t)) = \frac{q(\omega_t)}{p(\omega_t)}.
$$
Since the bet (amount placed at risk) is $B_t = 0.5M_t$ - the gambler stands to lose half of their wealth if  $\omega_t = 0$, since then the wealth is multiplied by $q(0)/p(0) = 1/2$.  
On the other hand, if $\omega_t = 1$, then the wealth is multiplied by $q(1)/p(1) = 3/2$.
In general, the payoff structure for the wager $B_t = 0.5M_t$ is $X_t(0) = -1$, and $X_t(1) = +1$.
This wager provides the following multiplicative updates to the wealth process for each coup
$$
(1+fX_t(\omega_t)) =
\begin{cases}
3/2 & \textrm{if}\  \omega_t = \textrm{1}, \\
1/2 & \textrm{if}\ \omega_t = \textrm{0}.
\end{cases}
$$
\cite{waudbysmith2022estimating} provide a universal representation for test martingales. 
We define the equation 
$$
f_t X_t(\omega_t) = \lambda_t (Y_t(\omega_t) - 0.5),
$$
where the fraction bet, $f_t$, is now dependent on time, and $Y_t(\omega_t) \in \{0,1\}$.
There is a one-to-one mapping from $f_t \in \{-1,1\}$ to $\lambda_t \in \{-2, 2\}$, where $f_t < 0$ is a bet in the opposite direction. 
In this form, the wealth process may be written as
$$
M_T = M_0\prod_{t=1}^T (1+ \lambda_t(Y_t - 0.5)),
$$
where we have suppressed the dependence of $Y_t$ on the outcome $\omega_t$.
This corresponds to setting $\lambda_t=1, $for all times $ t$. 
The gambler's wealth process, $M_t$, is a test martingale under the null hypothesis measure and by Ville's maximal inequality for nonnegative supermartingales,
$$
\mathbb{P}_{P}\left(\exists t \in \mathbb{N}~\mathrm{s.t.}~M_t \geq \frac{1}{\alpha}\right) \leq \alpha,
$$
for some $\alpha \in (0,1)$ \citep{Ville1939}.
Therefore, the gambler's wealth process will remain bounded by $1/\alpha$ with probability $\alpha$ if the null hypothesis is true (the coin is fair) and the Kelly betting system guarantees that the gambler's wealth will increase at a maximal rate (thus exceed the $1/\alpha$ level) if the alternative is true.

\subsection{Risk of Ruin}\label{sec:ruin}
The Kelly betting system provably achieves the maximal expected rate of return for the gambler's wealth if the alternative is true.
The gambler, however, may be unlucky. 
Indeed, \citet{kelly1956betting} remarked that if the betting game is stopped at some finite time $T$, then the optimal strategy depends on gambler's preference ``of being broke or possessing a fortune.''
Suppose that the alternative is true, $p=0.75$, but the outcome of the first three coin flips is $\omega^3 = (0,0,0)$.
Then the gambler's wealth is $M_3 = 1/8$ and it will take a larger number of samples to recover and finally exceed the $1/\alpha$ bound than if we started at $M_3=1$.
Though there is no risk of total ruin in proportional betting, some form of risk management is needed to prevent the gambler's wealth from declining to very low levels. Note that continuing the test from a wealth of $1/8$ corresponds to testing at much more stringent level of $0.00625$ on the remaining observations. An investigator with a limited budget may consider this scenario as ruin.

One way to mitigate the risk of ruin is to share the risk with other participants - individuals who seek risk.
Of course, sharing in the risk also comes with sharing in the rewards should the outcome be favorable. 
To accomplish this goal, we capitalize the wealth process and allow investigators to trade shares in that wealth process, a notion we formalize in the following section.

\section{Capitalizing the Testing by Betting Wealth Process}\label{sec:capitalizing}

We now formalize the notion of capitalizing and pricing the gambler's wealth process. 
Section~\ref{sec:simple_null_asset} motivates the fundamental of problem by finding a single price for a share in a Bernoulli experiment. 
Section~\ref{sec:characterizing} rigorously defines the gambler's wealth process in the testing by betting framework as an asset.
With these definitions, Section~\ref{sec:portfolio} defines an investigator's portfolio. We show that the investigator can tune the risk and return of their portfolio to their preference, and that the portfolio wealth process can be used in an anytime-valid sequential test.

\subsection{Simple Null and Alternative Bernoulli Experiment}\label{sec:simple_null_asset}

Continuing with the Bernoulli experiment from Section~\ref{sec:tech-prelim}, we can take the stochastic wealth process $M_t$ under the optimal betting system for the alternative to be an asset - that is, an uncertain sequence of payoffs starting from a fixed time $t$. 
The gambler generates $N$ shares of the wealth process such that each share is entitled to a right to a $1/N$ share of the equity of the wealth at any time $t$.

Prior to being capitalized, the asset is worth $M_0 = 0$ as the gambler has no wealth with which to bet on the outcome of the coin flips.
In the parlance of finance, the firm is nothing more than a betting strategy that may pay-off after future coin flips have happened.
The investigator will sell these shares to investors in return for capital that they will use to bet on the Bernoulli trials.
If these shares are to be sold to investors, what is the fair price for the shares?
To simplify this analysis, we will assume that the experiment will be terminated at $t=3$ and the shares will be liquidated and the capital will be returned to the investors.
This analysis can be extended to any finite time horizon.

Under the assumptions of no-arbitrage, the price of the asset is the present value of all future flows under the risk-neutral measure.
Assuming that there is no time-cost of capital, the discount rate is $d=0$.
Therefore, the price of the asset is $M_0 = \mathbb{E}_{\mathrm{risk-neutral}}[M_3]$.
In the testing by betting framework, the wealth process under the null is such that $M_0 = \mathbb{E}_P[M_t]$ for all $t$.
In the Bernoulli experiment, the measures are parameterized by a single parameter and the single equation constrains the risk-neutral measure to be equal to the measure under the null hypothesis.
This case illustrates the basic principle that we generalize in later sections.
Under the null, $p=0.5$, $\mathbb{E}_P[M_3] = M_0 = 1$.
This capitalizes the asset with the price of $M_0=1$.

Under the alternative $Q$, we have that $q=0.75$ and that
$$
\mathbb{E}_Q[M_3] = M_0 \left[q^3 \frac{27}{8} + 3 q^2 (1-q) \frac{9}{8} + 3 q (1-q)^2 \frac{3}{8} + (1-q)^3 \frac{1}{8}\right] \approx 1.953 M_0.
$$
If the coin flips are lucky, $(1,1,1)$, then the investors will be holding shares in an asset worth $\frac{27}{8}$ which is a 237.5\% increase in the value of their investment.
If the investors believe in the alternative, they expect to hold an asset worth $1.953$ which is an increase of 95.3\%.
Thus, the investors are induced to commit their wealth to the asset if they believe the alternative to be true, but they are unwilling to pay more than the value of the asset under the null because they are risk-neutral.

Both the investigator and the investors have risked a smaller portion of their fortune than if they were required to commit to fully capitalize $M_0$ themselves.
They, of course, do not enjoy the full extent of the upside, but they have mitigated their risk of ruin. Without a mechanism to share risk, the experiment may not have occurred if one investigator was required to fully fund the endeavor themselves. 

\subsection{Characterizing and Pricing the Wealth Process}\label{sec:characterizing}

We now formalize sequential testing notions in the framework of mathematical finance. 
We begin by defining the test wealth process, cash flow, and asset in Section~\ref{sec:definitions}. In Section~\ref{sec:asymp-cashflow} we consider the long run returns of the asset under the null hypothesis by analyzing the asymptotic behavior of the cash flow. We then turn to the problem of pricing the asset in Section~\ref{sec:pricing-asset} risk-neutral pricing. 

\subsubsection{Technical Definitions}\label{sec:definitions}
We first define a \emph{test wealth process}, which is closely related to the capital process of \citet{waudbysmith2022estimating}.

\begin{definition}[Test Wealth Process]\label{def:wealthprocess}
    Let a test wealth process be a stochastic process $K := (K_0, \ldots)$. The truncated process $K_{0,T} := (K_0, \ldots, K_T)$ is the stochastic process truncated at time $T$.
    The value of the wealth process at time $t$ is $K_t$. At time $t=0$, $K_0 = 1$.
\end{definition}

\begin{remark}
    If the test is of a simple null and simple or composite alternative, and we have that $\mathbb{E}_{P}(K_{t} \mid K_1,\dots,K_{t-1} ) = K_{t-1}$, then the test wealth process is equivalent to an e-process. 
\end{remark}

The science of investing can be broadly defined as the application of mathematical tools to tailoring a pattern of cash flows \citep{luenberger2014investment}. 
In mathematical finance, these cash flows represent a sequence of expenditures and receipts denominated in cash spanning over some time. 
A cash flow forms a defined, potentially random, sequence of flows into and out of a wealth process over time. 
We generalize this notion to the testing-by-betting wealth process and the flows into and out of that process due to the returns on individual unit bets.

\begin{definition}[Cash flow]\label{def:cashflow}
    Let a cash flow be a sequence $C = (c_1, \ldots$) such that $c_{t} = K_t - K_{t-1}$ for all $t>0$, with a final payoff denoting the price of the test wealth process at some stopping time. Like the wealth process, the cash flow process can be truncated as $C_{1,t} := (c_1, \ldots, c_{t})$. 
\end{definition}

In order to value and trade a cash flow, it is necessary to determine a fair price for a given cash flow by estimating the \emph{net present value}. 
We first broadly define an asset as a function of a cash flow, and then use a function of particular interest, the net present value function. 

\begin{definition}[Asset]\label{def:asset}
    An asset is a function of the cash flow sequence at time $t$, $g_t(C_{t,T})$. The net present value, $g_t^{\textrm{PV}}$, discounts future values of the cash flow such that $g_t^{\textrm{PV}}(C_{t,T}) = \mathbb{E}(\sum_{t'=t}^T d_{t,t'}c_{t'})$, where $d_{t,t'}$ is the discount factor from $t'$ to $t$, and the expectation is taken with respect to the true (possibly unknown) probability measure.
\end{definition}

The notion of discounting future values arises in mathematical finance from the time value of money. Since investors can earn interest risk free, the value of a dollar at future times is discounted by this risk free rate. If we assume that no arbitrage opportunities exist, then by the comparison principle, the value of an asset must discount future flows by this same rate \citep{luenberger2014investment}. Calculating the net present value is a widely used tool in mathematical finance, but it requires calculating an expectation with respect to the true underlying probability measure. If two separate investors have different beliefs of the underlying probability measure, then it is likely that their calculated net present values will differ as well. In order to trade the asset, a single price must be agreed upon.

\subsubsection{Asymptotic Behavior of Cash Flow under the Null Hypothesis}\label{sec:asymp-cashflow}

If $K_t$ is a test martingale, then it follows that the underlying cash flow excluding final sale is a \emph{martingale difference sequence}. There is a rich literature on characterizing the asymptotic behavior of martingale difference sequences. We now briefly present two results that suggest an investigator cannot expected to make a positive return on the test wealth process when the null hypothesis is true. 

Let us assume that these increments have a uniformly bounded second moment and that the null hypothesis is true. With these assumptions in mind, we now present the first result.

\begin{theorem}[Null Cash Flow Convergence]\label{thm:null-slln}
   Let $K_t$ be a test wealth process and let C be the corresponding cash flow increments, excluding final sale, with uniformly bounded second moments. If the null hypothesis is true, then as $T \to \infty$,  
   $$\frac{1}{T} \sum_{t=1}^{T} c_t \to 0 \hspace{1mm} \textrm{ almost surely. }$$
   
\end{theorem}
\begin{proof}
    Proof is provided in Appendix~\ref{appdx:slln-mds-proof}
\end{proof}
The result of Theorem~\ref{thm:null-slln} states that in the long run, a cash flow based on a test wealth process under the null hypothesis will on average bring a return of $0$ until final sale. We now turn to characterize the asymptotic distribution of such a process.

\begin{theorem} [Asymptotic Distribution of Cash Flow]\label{thm:cash-flow-clt}
    Let $K_r$ be a test wealth process and let C be the corresponding cash flow. If $T^{-1}\sum_{t=1}^{T} \mathbb{E}(c_t^2 | c_{1}, \dots, c_{t-1}) \xrightarrow[]{p} \sigma^2$, where $\sigma^2$ is a constant $0 < \sigma^2 < \infty$, and for any $\epsilon > 0$, $T^{-1} \sum_{t=1}^{T}\mathbb{E}[c_t^2 ~\mathbf{1} [| c_t | > \epsilon ] \mid c_1, \dots, c_{t-1}] \xrightarrow[]{p} 0$, then we have that
    
    $$\sqrt{T} \sum_{t=1}^T c_t \xrightarrow[]{d}N(0,\sigma^2).$$
\end{theorem}
\begin{proof}
This result follows directly from \cite{dvoretzky1972}.
\end{proof}
If the cash flow increments are uniformly bounded or have uniformly bounded second moments that do not vanish too quickly as $T \to \infty$, then the required conditions hold. 

We have characterized the long term mean return and variance under the null hypothesis. In essence, a return of mean 0 with some volatility is expected. This is a result that aligns with the anytime-valid inference literature, which states that a gambler betting against a true null hypothesis does not stand to make money in the long run.

\subsubsection{Pricing the Asset}\label{sec:pricing-asset}

We now consider the fundamental issue of pricing, as the previous net present value calculation in Section~\ref{sec:definitions} depends on the unknown true probability measure. In a no-arbitrage setting, assets are priced according to the discounted expected future cash flows, under the \emph{risk-neutral} measure. 
This measure is often referred to as the equivalent martingale measure since under this measure, the dynamics of the price of the underlying asset exhibit a martingale structure discounted by the risk free interest rate. The following proposition applies risk-neutral pricing to the testing by betting wealth process. 

\begin{proposition}[Risk-Neutral Pricing]\label{prop:risk-neutral-conservative}
    If the risk-neutral measure is indistinguishable from the measure under the null hypothesis, then any asset derived from a test wealth process has a risk-neutral price that is equivalent to the expected value under the null hypothesis.

\end{proposition}

\begin{proof}
   We call two sets \emph{equivalent} if they have the same null-set and we call two measures \emph{indistinguishable} if they assign an equal measure to all sets that do not have measure zero.
Proof is provided in Appendix~\ref{appdx:risk-neutral-pricing-valid}.
\end{proof}

We now show that for the running example Bernoulli experiment, the risk-neutral measure and the null hypothesis measure are indistinguishable.

\begin{proposition}[Underlying Measure]\label{prop:measure}
    Assume that the risk free interest rate is $0$. Suppose that we observe $(Y_t)_{t=1}^T \in \{0,1\}$ such that $\mathbb{E}(Y_t | Y_1, \dots, Y_{t-1}) = p$ and we define 
     $$K_t = \prod_{t=1}^T (1+\lambda_t(Y_t - p)),$$
     then the risk-neutral probabilities are identical to the probabilities under the simple null hypothesis, $H_0: p$, where $p \in [0,1]$.
     
\end{proposition}
Proposition~\ref{prop:measure} states that the distribution of $K_t$ under the null hypothesis is the distribution under the risk-neutral measure. This is verified in Appendix~\ref{appdx:proof-measure}.

The full utility of Propositions~\ref{prop:risk-neutral-conservative}~and~\ref{prop:measure} will become apparent in the following section. Briefly, these results imply that if assets are priced according to the risk-neutral price, then an investor who purchases these assets adopts a valid test wealth process, and in turn, Ville's inequality can be applied to conduct a valid sequential test. 

Risk-free interest rates, $r$, are typically nonnegative. 
We assume here that the risk-free interest rate is $0$. 
A positive interest rate would allow for a submartingale wealth process for an investigator with no risk, and in turn, monitoring the wealth of such an investigator would no longer be a valid sequential test. 
A negative interest rate would incentivize the investigator to invest in risky assets which may have a larger mean return and variance. 
We could interpret this negative rate as inflationary pressure to produce scientific results. 
However, if $r$ is negative and assets are priced according to the risk-neutral price, then the risk-neutral price will be less than the expected value under the null hypothesis. 
This implies that under the null hypothesis measure, the expected value of the asset is greater at a future time step, violating the necessary test martingale constraint for an anytime-valid test. 

\subsection{Constructing a Portfolio}\label{sec:portfolio}


We now turn our focus to the investigator and how they may build a portfolio consisting of shares of the asset $K_t$ and a risk free asset. By controlling the allocation of their wealth between a risk free and a risky asset, the investigator can tailor the mean and variance of the portfolio's return to their risk preference. This connection with mean-variance portfolio optimization allows the investigator to reduce the risk of ruin when investing in $K_t$ \citep{markowitz1952portfolio}. 
The impact of the portfolio approach in a statistical hypothesis setting is that an investigator may take longer to reject the null hypothesis if it is truly false, but they gain by preserving wealth for other tests that may be more beneficial. 
We now formally define such a portfolio and show that we can construct a valid test based on its value.

\begin{definition}[One-Fund Testing Portfolio]\label{def:portfolio}
    Let a one-fund testing portfolio be a two-element vector of wealth processes $\mathbf{K}_t := (\free_{t}, \risky_t)$, where $\free_{t}$ denotes the value of the risk free wealth process and $\risky_t$ denotes the value of the shares owned of the risky test wealth process.
\end{definition}

\begin{lemma}\label{lemma:one_fund}
    Assume that the risk free interest rate is $0$ and that the risk-neutral measure is indistinguishable from the null hypothesis measure. Assume that we permit the investigator to short the risky asset up to $\frac{\free_t + u \risky_t}{u-1}$ or to take a one-period loan up to $\frac{\free_t + d \risky_t}{1-d}$, where $u$ and $d$ denote bounds on the multiplicative upside and downside of $\risky_t$, respectively. Then the total value of $\mathbf{K}_t$ is a test martingale under the null hypothesis and $\lvert \mathbf{K}_t \rvert \geq \frac{1}{\alpha}$ is an anytime-valid level $\alpha$ test for some $\alpha \in (0,1)$.
\end{lemma}
\begin{proof}
    Proof is provided in Appendix~\ref{appdx:one-fund-proof}.
\end{proof}
Definition~\ref{def:portfolio} and Lemma~\ref{lemma:one_fund} are similar to the Sequentially Rebalanced Portfolio strategy of \citet[Eq.~48]{waudbysmith2022estimating}. Their definition of a Sequentially Rebalanced Portfolio assumes a convex combination of test (super)martingales, whereas Definition~\ref{def:portfolio} does not explicitly assume that $\free_t$ and $\risky_t$ are test (super)martingales. Lemma~\ref{lemma:one_fund} permits $\free_t$ and $\risky_t$ to take negative values. Even with this relaxed assumption, we show that the resulting wealth process, the value of the portfolio, is still a test martingale. A negative value for $\free_t$ occurs when the investor takes a loan to purchase additional shares of $K_t$, which is referred to as being ``long'' on $K_t$, while a negative value for $\risky_t$ occurs when an investigator short sells shares of $K_t$. Short selling is done by borrowing shares and then immediately selling the shares, with the hope that in the future the investor may repurchase the shares at a lower price to return to the lender. Short selling is particularly useful when the investor believes that the price of the asset will decline, for example, when the price of the asset is a supermartingale. With this flexibility, the investor may tailor their portfolio to their risk preference. This is particularly useful when the investor is not in control of the betting strategy, but also can be useful when they do have control of the betting strategy since they do not have to explicitly calculate a new betting strategy. 

\begin{remark} Assume that the mean rate of return of the risky asset is $\mu_{K_t} \geq 0$ with standard deviation $\sigma_{K_t} > 0$. If short selling or loans are not permitted, then the investigator can construct a portfolio $\mathbf{K}_t$ with return parameters $\Tilde{\mu}, \Tilde{\sigma}$ such that $\mu_{K_t} \geq \Tilde{\mu}\geq 0$ and $\sigma_{K_t} \geq \Tilde{\sigma}\geq 0$. Furthermore, if we permit loans and short selling up to the limits of $\frac{\free_t + u \risky_t}{u-1}$ or to take a one-period loan up to $\frac{\free_t + d \risky_t}{1-d}$, where $u$ and $d$ denote bounds on the multiplicative upside and downside of $\risky_t$, respectively, then the investigator can construct a portfolio $\mathbf{K}_t$ with return parameters $\Tilde{\mu}, \Tilde{\sigma}$ such that $\Tilde{\mu} \geq \mu_{K_t} \geq 0$ and $ \Tilde{\sigma}\geq \sigma_{K_t} \geq 0$ or $\mu_{K_t} \geq \Tilde{\mu} \geq 0$ and $  \sigma_{K_t} \geq \Tilde{\sigma}\geq 0$.
\end{remark}

Because the portfolio construction allows for loans, the wealth process of the portfolio can replicate any other test wealth process with a betting strategy that corresponds to a different alternative hypothesis.
Since the test wealth process is a test martingale under the null, it follows from the Universal Representation theorem of test martingales due to \citet[Prop.~3]{waudbysmith2022estimating} that the value of $\mathbf{K}_t$ can be written in the form of a distinct, but correlated, asset. This representation can be recomputed at each time step, and then the resulting strategy may be written in the form of a Sequentially Rebalanced Portfolio \citep[Eq.~48]{waudbysmith2022estimating}. 

\section{Derivatives of the Testing by Betting Wealth Process}\label{sec:derivatives}

In this section we consider pricing derivatives of assets so that the investigator can hedge certain risks. Through a demonstrative example in Section~\ref{sec:bernoulli-derivative}, we first establish that risk-neutral pricing for derivative contracts enables time-uniform error control for sequential tests based on portfolios containing these contracts. In Section~\ref{sec:derivative-portfolio} we consider the construction of arbitrary portfolios with derivatives contracts. We then consider pricing derivative contracts for an experiment with Normally distributed outcomes in Section~\ref{sec:derivative-normal}. This section motivates the need to be able to approximate the price of derivative contracts, and we discuss the use of Monte Carlo simulation for approximating the price of derivatives similar to European option contracts is Section~\ref{sec:monte-carlo}.

\subsection{Pricing Derivatives in Bernoulli Experiments}\label{sec:bernoulli-derivative}

Derivatives are assets that derive their values from the wealth process of some other underlying asset.
Mathematically, the original asset's cash flow derives from the random outcome of a sequence of experiments and a particular betting strategy.
As was shown in Section~\ref{sec:characterizing}, the present value is a function of this cash flow process. 
A derivative is then simply a (typically nonlinear) function of the assets and therefore a function of a function of the cash flow.

To fix these ideas, suppose that, at $t=0$, a call option on the wealth process of the asset that bets optimally according to the belief that $p=0.75$ provides the investigator the right, but not the obligation, to purchase 1 share of the asset at $t=3$.
What is the fair value of this call option?

We leverage risk-neutral pricing over a binomial lattice for this problem. We assume that the risk-free interest rate is $r=0$. The risk-neutral price of an option contract is the discounted expected value of the contract at expiry with respect to the risk-neutral measure. Assuming a strike price of $S$ at expiry time $\tau$, the value of the call option is $K_\tau^{\mathrm{call}} :=\max\{0, K_{\tau} - S \}$. The underlying asset can either increase or decrease by a factor of $\frac{1}{2}$ at each time. The risk-neutral probabilities are then $(0.5, 0.5)$ for an increase and decrease respectively, a result implied by Proposition~\ref{prop:measure}.This expectation can be taken across multiple time periods. 

Such a contract offers a risk-seeking speculator the option to purchase shares of $K_t$ in the future at a discounted rate in the event that favorable results occur. 
At the same time, a risk-averse hedger can use the option contract to protect against ruin. 

We illustrate an example of the pricing of a speculative option contract in 
Figure~\ref{fig:put-price},
where we price a put option contract with expiry $\tau = 3$, strike price $S = 1/4$ and $K_0 = 1$. 
Since $S < K_0$, this option is more valuable for the speculator than for the hedger.
At expiry the value of the option is non-zero only if the outcomes of the Bernoulli experiment is $0$ at all three time periods.
The price of the option is only $1/64$, so for a price of less than $K_0 = 1$ the speculator can take a stake in the random sequence that pays $1/8$ after exercising the option and liquidating the shares.
The hedger is then left with a wealth of $1/4$ instead of $1/8$.



\begin{figure}[!htbp]
    \centering
    \includegraphics[width=0.6\textwidth]{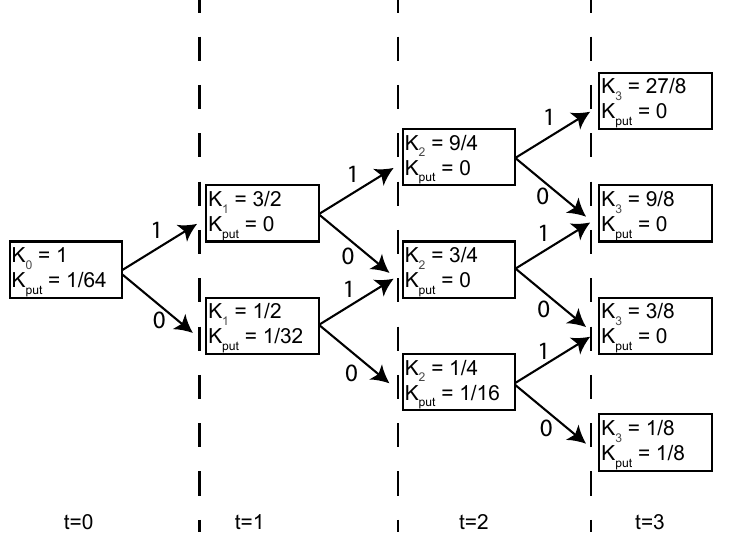}
    \caption{Calculating the price of a call option with expiry of $\tau = 3$, with a strike price of $S = 10/8$ at time $t=0$. The value of the call option at (i) follows from $K_{call} = \max\{0,K_3 - S\}= 27/8 - 10/8$. At (ii) the expected value of the contract is $0.5 (17/8) + 0.5 (0) = 17/16$. A similar process is repeated at (iii) and (iv) to have a price at time $t = 0$ of $17/64$. By only risking $17/64$, the investigator has the opportunity to purchase the risky asset a discounted price if favorable results occur.}
    \label{fig:put-price}
\end{figure}

\subsection{Portfolios with Derivative Contracts}\label{sec:derivative-portfolio}

With the addition of the derivative contract, we augment the investigators portfolio to contain three elements, $\mathbf{K} := (\free_{t}, K^{\mathrm{risky}}_t, K^{\mathrm{deriv}}_t)$, where $K_t^{\mathrm{deriv}}$ denotes the value of the derivative contracts held by the investigator at time $t$. 
\begin{lemma}\label{lemma:options-portfolio}
    Assume that the risk free interest rate is 0 and that the risk-neutral measure is indistinguishable from the null hypothesis measure. Let an investigator's portfolio contain three elements, a risk free asset, a risky asset, and a derivative contract, defined as $\mathbf{K}_t := (\free_{t}, K^{\mathrm{risky}}_t, K^{\mathrm{deriv}}_t)$. Assume that we permit the investigator to short the risky asset, take a loan, or purchase or issue a derivative contract such that probability of negative wealth is $0$. Then the total value of $\mathbf{K}_t$ is a test martingale under the null hypothesis and $\lvert \mathbf{K}_t \rvert \geq \frac{1}{\alpha}$ is an anytime-valid level $\alpha$ test for some $\alpha \in (0,1)$.
\end{lemma}
\begin{proof}
    Proof is provided in Appendix~\ref{appdx:deriv-proof}.
\end{proof}


The result of Lemma~\ref{lemma:options-portfolio} allows the investigator to augment the previously described portfolio with derivative contracts. This gives a risk controlling device that is easy to design and interpret, and can perfectly hedge against ruin. For example, a put option may be purchased to prevent the ruin described in our motivating example of observing $(0,0,0)$. This risk control can be accomplished without any estimation of the rate of return or volatility of an asset and does not invalidate statistical guarantees of the testing process. Although we only explicitly consider European put and call options in this work, \emph{any} derivative contract based on $\risky_t$ which is traded at the risk-neutral price can be used in our portfolios.  

Before proceeding, we pause to note that a portfolio containing an option can be replicated by a portfolio containing only a risk free and a risky asset. Such a portfolio is called a replicating portfolio. It follows that such a portfolio could be represented by a version of the Sequential Rebalanced Portfolio \citet[Eq.~48]{waudbysmith2022estimating}. However, it is important to point out how simple it is to construct and trade the option, or other derivatives in general, in comparison to a betting scheme which satisfies these risk controlling constraints. Indeed, it is much simpler to construct a replicating portfolio once the price of the option is determined.

\subsection{Pricing Derivatives with Normal Outcomes}\label{sec:derivative-normal}

Suppose we observe $Z_1, \dots, Z_T$ where $Z_t \stackrel{iid}{ \sim} N(\mu,1) $. We wish to test the null hypothesis $H_0: \mu = 0$ versus the alternative hypothesis $H_1: \mu = \mu_1$ for some $\mu_1 > 0$. Consider the transformation $Y_t = e^{Z_t}$. It follows that $Y_t$ follows a log-normal distribution with parameters $\mu$ and $1$, and $\mathbb{E}_{P}(Y_t) =  e^{1/2}$, where $P$ is the measure under the null hypothesis. We can construct a test wealth process which is a test martingale under the null hypothesis using the Universal Representation Theorem \citep{waudbysmith2022estimating}. This wealth process can be expressed as
\begin{equation}\label{eq:lognormal-capital}
    K_T = \prod_{t=1}^T(1 + \lambda_t(Y_t - e^{\frac{1}{2}})).
\end{equation}

Appendix~\ref{appdx:normal_martingale} shows that this process is a test martingale under the null hypothesis.
If $\lambda_t = e^{-\frac{1}{2}}$, then we have that $(1 + \lambda_t(Y_t - e^{\frac{1}{2}}))$ is a log-normal random variable with parameters $\mu + \log \lambda_t$ and $1$. The product of independent log-normal random variables is log-normal, and hence $K_T \sim \mathrm{log-normal}\left( \sum_{t=1}^{T} (\mu + \log \lambda_t), T \right)$. The measure under the null hypothesis satisfies the necessary conditional expectation condition $\mathbb{E}(K_T | K_1,\ldots, K_{T-1}) = (\frac{1}{1+r})K_{T-1} $ where $r = 0$ by assumption. Since the only free parameter is $\mu$ and we have placed a constraint on the expected value of $K_T$, the risk-neutral measure is indistinguishable from the measure under the null hypothesis. Hence, risk-neutral pricing may be used to price an option so that an investigator's wealth process remains test martingale. This is simply applying the same logic as Proposition~\ref{prop:measure} to an experiment with normal outcomes. 

Note that the test wealth process is the product of independent log-normal increments whose variance is proportional to the distance between observations. If we look at this process from a large time scale, meaning that $T$ grows large, then the distance between increments will appear small. At this scale, the test wealth process would closely resemble Geometric Brownian Motion. With this in mind, we can approximate the price of European options with the Black-Scholes equations \citep{black1973scholes}. The price of a European call option purchased at time $t$ with strike price $S$ and expiry $\tau$ can be approximated as

\begin{equation}
    K_t^{\textrm{call}} = K_t \Phi\left(\frac{K_t}{S} + \frac{\sigma^2 (\tau - t)}{2} \right) - S \Phi\left( \frac{K_t}{S} + \frac{\sigma^2 (\tau - t)}{2} - \sigma \sqrt{\tau -t}\right).
\end{equation}

In order to obtain a known distribution for $K_T$, we have subtly assumed that the betting strategy follows an all or nothing protocol, which at each time step, assigns a non-zero probability of total ruin. If we allow $\lambda_t \in [0,e^{-\frac{1}{2}})$, then we have that $\lambda_tY_t - (1 - \lambda_t e^{-\frac{1}{2}})$ is a log-normal random variable. Recovering the exact distribution of $K_T$ is not feasible at scale, however, it can be approximated through Monte Carlo simulation. 

\subsection{Pricing Derivatives via Monte Carlo}\label{sec:monte-carlo}

We now turn our attention to approximating the price of derivatives through Monte Carlo simulation. We assume that $Y \stackrel{iid}{\sim} F_y$, where $F_y$ denotes the cumulative distribution function of $Y$ under the null hypothesis. As we have done in previous sections, we assume that the risk-neutral measure is indistinguishable from the measure under the null hypothesis. We can define $\deriv_t = g(Y_t,\dots,Y_\tau)$, meaning that the price of the contract is given as a function of future random observations. If we simulate $n$ sample paths of $(Y_t,\dots,Y_\tau)$, then we can use the sample average $\widehat{\deriv_t} = \frac{1}{n} \sum_{i = 1}^n g(Y^i_t,\dots, Y^i_\tau)$ to approximate $\mathbb{E}(g(Y_t,\dots, Y_\tau)) =\deriv_t$.

This estimate is unbiased, so we have that $\mathbb{E}(\widehat{\deriv_t}) = \deriv_t$. This implies that using the estimate in place of the true risk-neutral price allows the investigator to trade this contract without invalidating time-uniform error guarantees. By the weak law of large numbers, we also have that $\widehat{\deriv_t} \xrightarrow[]{p} \deriv_t$ as the number of Monte Carlo simulations goes to $\infty$. If we desire an estimate with lower variance, increasing the number of Monte Carlo samples will achieve this goal.

\section{Empirical Results}\label{sec:empirical-results}

We now present empirical results that demonstrate the practical benefits of using derivatives to hedge against ruin in testing-by-betting. We provide two simulation studies. In Section~\ref{sec:bernoulli-experiment} we continue with a simple null versus a simple alternative in a Bernoulli experiment, and allow the investigator to purchase a European put option to prevent ruin. In Section~\ref{sec:drift-experiment}, we modify the Bernoulli experiment to have a distribution shift midstream. In Section~\ref{sec:real-data} we consider anytime-valid testing of differential gene expression data \citep{dettling2004bagboosting}. 

\subsection{Bernoulli Experiment}\label{sec:bernoulli-experiment}
We now simulate the motivating example of testing the hypothesis of $H_0: p = 0.5$ versus $H_1: p = 0.75$ for independent and identically distributed $Y_1,\dots, Y_{20}  \sim \textrm{Bernoulli}(p)$. We define $\Tilde{K_t}$ as the set of wealth processes which have not exceeded the $1/\alpha$ threshold by time $T = 20$, which implies that the corresponding testing process has failed to reject the null by that time. Explicitly, $\Tilde{K_t} := \{K_t~\textrm{s.t.}~(\max K_{0:20}) < \frac{1}{\alpha}\}$. We define the $0.01$-quantile of $\Tilde{K_t}$ as $k_{0.01}$. We denote the expected value of $k_{0.01}$ as $\mathbb{E}(k_{0.01})$. This value is related to the conditional value at risk, which would be the expected loss, $\mathbb{E}(1 - k_{0.01})$. We estimate $k_{0.01}$ with the sample quantile of the simulated wealth processes in $\Tilde{{K_t}}$. We denote the set of wealth processes in $\Tilde{K_t}$  with a final wealth below $k_{0.01}$ as $\Tilde{K_t}_{,0.01}$. The expected value is approximated with the sample mean from the Monte Carlo simulations.

In this simulation study, we consider ruin to be $K_{20} < 0.25$ and that $K_t$ has not exceeded $1/\alpha$ at any prior time. We can pose this as a risk control problem where we wish to have $\mathbb{E}(k_{0.01}) \geq 0.25$. To control risk, we consider purchasing a European put option with expiration $\tau = 20$. We would like to purchase an option at strike price $S$ such that in the event of exercise, we will have $K_{20} = 0.25$. The strike price must be $0.25 = (1-C)*S,$ where $C$ denotes the purchase price of the option contract. The solutions to the previous equation are $S = 0.30866$ or $S = 0.97285$, and we choose to select the lower priced option, with a lower strike price to engage more of our wealth in the Kelly betting system on the risky asset. We compare these strategies to \emph{Kelly} betting without purchase of the option, a \emph{conservative} fixed $\lambda = 0.133934$ which the worst case scenario has $K_{20} = 0.25$. Finally, we consider a \emph{dynamic} strategy with $\lambda_t~\textrm{s.t.}~0.25= K_t (1 + \lambda_t(-0.5))^{20-t},$ which gives 
$$
\lambda_t = 2\left(1 - e^{(20-t)^{-1} \ln\left(\frac{.25}{K_t}\right)}\right).
$$
This dynamic strategy calculates the value of $\lambda_t$ that would have a worst case $K_{20} = 0.25$ given the current wealth and time.

\begin{table}[]
    \centering
    \begin{tabular}{lrrrr}
    \toprule 
    &       Kelly &  Conservative &    Dynamic &    Option  \\
    \midrule
    Avg. ${K}_{20}$        &   89.44 &      1.94 &  12.43 &  72.44 \\
    Power            &    0.53 &      0.00 &   0.22 &   0.51  \\
    Avg. $\max{K_t}$          &  108.04 &      2.00 &  15.98 &  87.50  \\
    Avg. $K_{20} \in \tilde{K}_{20}$    &    5.21 &      1.94 &   4.29 &   4.41  \\
   $\mathbb{E}(k_{0.01})  $            &    0.043 &      0.904 &   0.250 &   0.250  \\
    \bottomrule \\
    \end{tabular}
    \caption{Results of $10,000$ iterations for four investing strategies. The Kelly betting strategy results in the largest wealth and power, but risks ruin. All other strategies prevent ruin, but only the options based strategy is competitive to the Kelly betting system in terms of final wealth and power.}
    \label{tab:bernoulli_results}
\end{table}

Table~\ref{tab:bernoulli_results} shows results from $10,000$ repetitions. The Kelly betting strategy shows the greatest power, expected final wealth, and the greatest expected wealth among tests which have not yet rejected. These benefits come at the risk of ruin, as seen in $\mathbb{E}(k_{0.01})$. The expected final wealth in the $1\%$ tail is $0.043$, which is far below the desired threshold of ruin at $0.25$. All other investment schemes control this risk. The conservative and dynamic schemes achieve risk control, with certainty, by sacrificing power. The strategy that purchases an option and invests the remaining wealth into the Kelly betting strategy satisfies our risk control criterion, while still remaining comparable to the optimal Kelly betting strategy. The option strategy enjoys similar power to the Kelly betting strategy, as most capital is invested in the Kelly betting strategy while the risk of ruin is provably eliminated.

Figure~\ref{fig:bernoulli-hist} shows the empirical distributions of the the wealth processes at time $T=20$. The strategy based on combining Kelly betting with a put option has a similar distribution to the Kelly betting strategy but has no risk of experiencing ruin. This comes at the cost of slightly lower wealth on average. 
But, recall the goal of testing-by-betting  is not to accrue as much wealth as possible, but to correctly reject the null hypothesis - any wealth beyond the rejection point confers no additional utility.

\begin{figure}
    \centering
    \includegraphics[width = .9\textwidth]{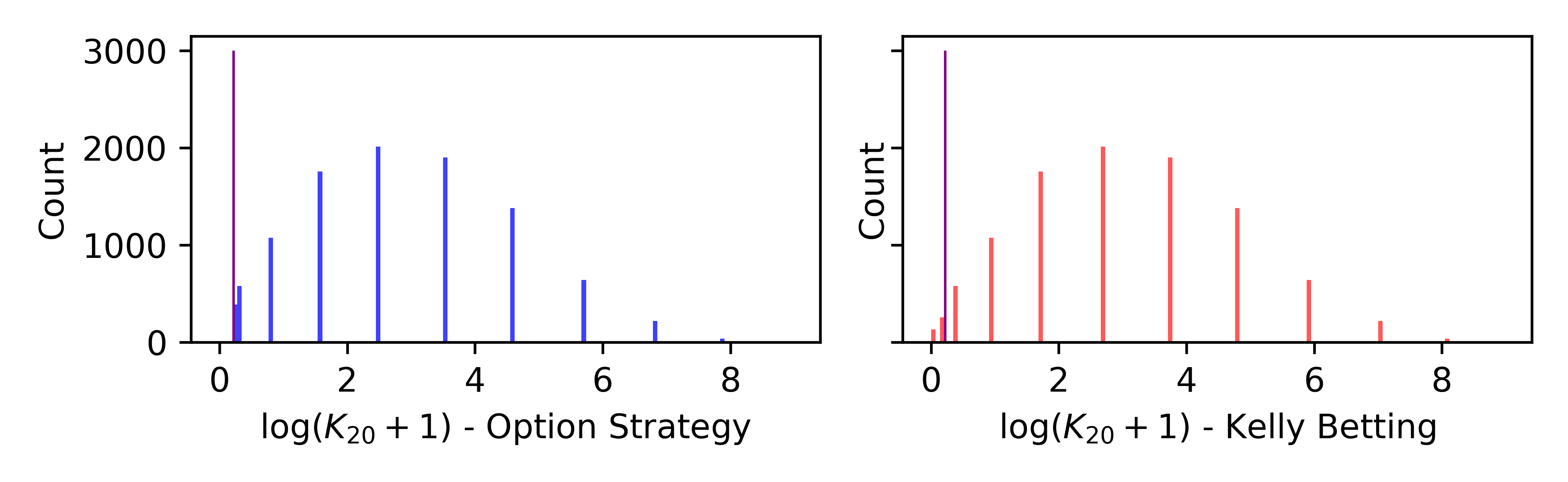}
    \caption{Empirical Distribution of $K_{20}$ following various investing strategies. The Kelly betting system has a non-zero mass to the left of $0.25$, where as all other methods do not. The strategy with a European put option follows a similar distribution to the Kelly betting system.}
    \label{fig:bernoulli-hist}
\end{figure}

\subsection{Bernoulli Experiment Sequence with Distributional Shift}~\label{sec:drift-experiment}

We now modify the previous simulation to consider the utility of options when there is a distribution shift during testing. 
Such an phenomena may occur when there is a delay in the effect of treatment. 
We simulate $T=20$ data points, where $Y_1,\dots,Y_{10} \sim \textrm{Bernoulli}(p = 0.5)$ and $Y_{11},\dots,Y_{20} \sim \textrm{Bernoulli}(p = 0.75)$. Our hypotheses are modified to be
$$H_0 : p = 0.5~\forall t~~\textrm{vs}~~H_1 : \exists t~\textrm{s.t.}~p = 0.75~\forall t^{*} \geq t .$$

Although our alternative hypothesis is now composite, the null hypothesis is still simple. The null hypothesis is identical to that of the previous simulation, and hence the previously considered test wealth process remains a test martingale. This simulation set up is inspired by that of \citet[Section~3.4]{chugg2023fairness}. We consider ruin to remain at $K_{20} = 0.25$, and we elect to purchase the option with a lower contract and strike price. However, we now allow for expiry to occur at either $\tau =10$ or at $\tau =20$. Note that $\tau=10$ corresponds to pricing an option for the \emph{exact} time at which the shift occurs, while $\tau=20$ is an option that will provably hedge against the risk of ruin at time $T=20$. The risk of ruin is compounded by sub-optimal betting during the first $10$ samples. 

Table~\ref{tab:bernoulli_drift} gives results for $10,000$ repetitions. By design, the conservative, dynamic, and $\tau=20$ strategies guarantee that the final wealth does not fall below the prespecified level of ruin. With a well-timed put option at time $\tau = 10$, power can be increased relative to the misspecified Kelly betting strategy. However, the expected value of the final wealth in the left tail is improved by a factor of $4$. When the put option is purchased with an expiry at $\tau = 20$, a single additional sample can recover the wealth attained by Kelly betting in the average case. However, for the cases in the left tail, a minimum of 12 samples will be required in order to recover the wealth at $T=20$ of the $\tau=20$ option strategy. In this scenario, an investigator will have to spend far more resources in order to eventually reject the null hypothesis and the investigator may simply abandon the experiment.

\begin{table}[]
    \centering
    
    \begin{tabular}{lrrrrr}
    \toprule
    {} & {} & {}  & {} & \multicolumn{2}{c}{Options} \\
    {} &      Kelly &  Conservative &   Dynamic &    $\tau=10$ &   $\tau=20$ \\
    \midrule
    Avg. ${K}_{20}$         &   9.63 &      1.39 &  5.69 &   9.38 &  7.85 \\
    Power             &   0.09 &      0.00 &  0.07 &   0.11 &  0.09 \\
   Avg. ${\max{K_t}}$          &  12.03 &      1.46 &  7.18 &  11.56 &  9.74 \\
    Avg. $K_{20} \in \Tilde{K}_{20}$  &   2.43 &      1.39 &  3.60 &   2.61 &  2.05 \\
    $\mathbb{E}(k_{0.01})  $   &   0.002 &      0.631 &  0.250 &   0.010 &  0.250 \\
    \bottomrule\\
    \end{tabular}
    \caption{Results of $10,000$ iterations with distribution shift. The Kelly betting strategy results in the largest wealth and power, but risks ruin. All other strategies prevent ruin, but only the options based strategy are competitive to the Kelly betting system in terms of final wealth and power.}
    \label{tab:bernoulli_drift}
\end{table}





\subsection{Gene Expression Data}\label{sec:real-data}

Microarray data for $m=6033$ gene expression levels across 50 normal samples versus 52 prostate cancer tumor samples was sequentially tested~\citep{dettling2004bagboosting}. The data set was normalized and log-transformed. Under the null hypothesis, each gene follows a standard normal distribution. The data was then transformed using the standard normal CDF so that data is bounded between $0$ and $1$. Following this transformation, the data follows a Uniform$(0,1)$ distribution under the null hypothesis, with mean $\mathbb{E}(Y_i) = 0.5$. The goal of each individual test is to determine if gene expression is significantly different for tumor samples in comparison to the normal samples.

Let $Y_i$ denote the transformed expression level for the $i^{\textrm{th}}$ gene. The null hypothesis is $H_0: \mathbb{E}(Y_i) = 0.5$ and the alternative is $H_1: \mathbb{E}(Y_i) \neq 0.5$. We consider ruin as a wealth of $0.5$, and price a put option with expiry at time $\tau = 50$, at the middle of testing process. We price this put option using Monte Carlo simulations. We use (and subsequently discard) two samples from the tumor samples subjects to estimate a prior probability of the alternative hypothesis \citep{cook_dubey_lee_zhu_zhao_flaherty_2024}. This prior is then used to determine the betting fraction, $\lambda$, to be used in the sequential test. The underlying capital process follows that of the Hedged-CS \citet{waudbysmith2022estimating}. Explicitly
the process is 
$$K_t = 0.5 \prod_{t=1}^T (1+\lambda(Y_t - 0.5)) + 0.5 \prod_{t=1}^T (1-\lambda(Y_t - 0.5)).$$

\begin{table}[]
    \centering
    
    \begin{tabular}{rrr}
    \toprule
    {} & {Kelly} & {Option $\tau = 50$} \\
    \midrule
    Avg. ${K}_{20}$    &  1.71e10&      8.61e+9 \\
    Prop. Rejected      &   .340 &      .310 \\
    Avg. $K_{20} \in \Tilde{K}_{20}$  &   3.483&      3.436  \\
    $\mathbb{E}(k_{0.01})$   &   .003 &    0.5 \\
    \bottomrule\\
    \end{tabular}
    \caption{Simulated sequential testing of Microarray data. Purchasing a put option provably controls risk with small degradation to the probability of rejection.}
    \label{tab:real_data}
\end{table}

Table~\ref{tab:real_data} shows the results for performing sequential tests for all $6033$ genes. The Kelly betting system achieves a higher average final wealth and has a larger proportion of hypotheses rejected. However, the expected wealth in the $1\% $ is just $0.003$, whcih corresponds to a return of $-99.7\% $ return on investment. As intended, purchasing a put option prevents ruin. This comes at the price of being able to reject hypotheses, as was seen in the synthetic data experiments.

\section{Conclusion}

We have characterized the stochastic process of a gambler's wealth in the testing by betting framework as a stochastic process that represents the price of an asset. We have constructed an investigator's test wealth process as the value of a portfolio between a risk free and risky asset. By allowing the investigator to trade these assets at the risk-neutral price, the investigator's wealth process can be used to conduct an anytime-valid sequential test. An investigator can perfectly hedge against the risk of ruin by purchasing forward-looking derivative contracts, which are the first of their kind in the testing by betting literature. Empirical results show that the risk of ruin can be completely eliminated while maintaining competitive power to the long-term optimal Kelly betting system.

\bibliography{references}

\appendix

\section{Proof of Theorem~\ref{thm:null-slln}}\label{appdx:slln-mds-proof}

 \begin{proof}
    By assumption $\mathbb{E}({C_t}^2)$ is uniformly bounded by some constant $b < \infty$. For a martingale difference sequence with $\sum_{t=1}^{\infty}\mathbb{E}_P(C_t^2)/t^2 \leq \sum_{t=1}^{\infty}b/t^2 \leq \infty$, we have that by the strong law of large numbers for martingales, $T^{-1} \sum_{t=1}^{T} c_t \to 0 $ almost surely. 
\end{proof}

\section{Proof of Proposition~\ref{prop:risk-neutral-conservative}}\label{appdx:risk-neutral-pricing-valid}
\begin{proof}
    We have that under the null hypothesis, the expected value of a cash flow derived from a test wealth process, whose elements could be random variables, is expressed as,
  \begin{equation*}  
    g_{t}^{\textrm{PV}} = \mathbb{E}_P\left(\left(\sum_{t'=t}^T d_{t,t'}c_{t'}\right) + K_T\right)=  \mathbb{E}_P\left(\sum_{t'=t}^T c_{t'}\right) +\mathbb{E}_P \left(K_T\right) = \sum_{t'=t}^T \mathbb{E}_P \left(c_{t'}\right) + K_0 = 1
   \end{equation*} 
Under risk-neutral pricing, the test wealth process evolves such that $\mathbb{E}_{\textrm{risk-neutral}}(K_{t}) = \frac{1}{1+r} K_{t-1} = K_{t-1}$. Hence, the increments also have expectation $0$, and by we have that $\mathbb{E}_{\textrm{risk-neutral}}(K_T) = K_0 = 1$. 
\end{proof}

\section{Proof of Lemma~\ref{lemma:one_fund}}\label{appdx:one-fund-proof}

\begin{proof}

Let us first consider the case that negative values for $\free_t$ or $\risky_t$ are not permitted. To establish that $\lvert \mathbf{K}_t \rvert$ is a test martingale under the null hypothesis, we must show that the the initial value $\lvert \mathbf{K}_0 \rvert = 1$, that $\lvert \mathbf{K}_t \rvert$ is nonnegative for all $t$, and that $\mathbb{E}_{P} (\lvert \mathbf{K}_t \rvert \mid K_{1},\dots,K_{t-1}) = K_{t-1}$, where $P$ is the measure under the null hypothesis.
By definition the initial portfolio value is 
$$\lvert \mathbf{K}_0 \rvert = \free_0 + \risky_0 = 1 + 0 = 1.$$

Next we verify nonnegativity. Since
$\lvert \mathbf{K}_t \rvert = \free_t + \risky_t,$ we can show that $K_t$ is nonnegative by showing $\free_t$ and $\risky_t$ are nonnegative. Note that since short selling is not allowed, this means that any time point, the amount that $\free_t$ can decrease (be invested in $\risky_t$) is at most $\free_t$. This implies that $\free_t$ is nonnegative. Note that $\risky_t =  \frac{n_{\textrm{shares}, t}}{N} M_t$ which is a product of nonnegative terms, and hence $\risky_t $ is nonnegative. 

Lastly, we establish $\mathbb{E}_{P} (\lvert \mathbf{K}_t \rvert \mid K_1 , \dots, K_{t-1}) = K_{t-1}.$
\begin{align*}
    \mathbb{E}_{P} (\lvert \mathbf{K}_t \rvert \mid K_1 , \dots, K_{t-1}) &= \mathbb{E}_{P}(\free_t + \risky_t \mid K_1 , \dots, K_{t-1})\\
    &= \free_{t-1} + \risky_{t-1}\mathbb{E}_{P}(1 + fX_l(\omega_l) | K_1 , \dots, K_{t-1}) \\
    &=  \free_{t-1} + \risky_{t-1}\mathbb{E}_{P}\left(\frac{q(\omega_l)}{p(\omega_l)} \Bigm| K_1 , \dots, K_{t-1}\right)\\ &= \free_{t-1} + \risky_{t-1} = \lvert \mathbf{K}_{t-1} \rvert.
\end{align*}

Hence $\lvert \mathbf{K}_t \rvert$ is a test martingale. Note that we need only consider the expected change in the portfolio value following an observation, as any rebalancing of the portfolio between observations does not affect the value of the portfolio itself.

Now we consider the case that short selling is permitted and that loans are permitted with an interest rate of $0$. By definition the initial value remains unchanged. Let the worst case multiplicative downside of $\risky_t$ be denoted as $d$. Assume that a single-period loan of $a$ dollars is issued for purchasing the risky asset. Then at time $t$, the portfolio is rebalanced from $[\free_{t}, \risky_t]$ to $[\free_{t} - a, \risky_t + a]$. It then follows that the portfolio at time $t+1$ is $[\free_{t} - a, d(\risky_t + a)]$. If $ 0 \leq a \leq \frac{\free_t + d \risky_t}{1-d}$, then we have that $\lvert \mathbf{K}_{t+1} \rvert$ is given by,
\begin{align*}\lvert \mathbf{K}_{t+1} \rvert &= \free_{t} - a +  d(\risky_t + a)\\
&\geq \free_{t} + d(\risky_t) + \frac{\free_t + d \risky_t}{1-d}(d-1)\\
&= \free_{t} + d(\risky_t) -\free_{t} - d(\risky_t) = 0.
\end{align*}

Using an induction argument, $\lvert \mathbf{K}_t \rvert$ is nonnegative for all $t$.

Lastly, 
\begin{align*}
    \mathbb{E}_{P} (\lvert \mathbf{K}_t \rvert \mid K_1, \dots, K_{t-1}) &= \mathbb{E}_{P}(\free_t - a + (\risky_t+a) \mid K_1, \dots, K_{t-1})\\
    &= (\free_{t-1}-a) + (\risky_{t-1} + a)\mathbb{E}_{P}(1 + fX_l(\omega_l)|K_1, \dots, K_{t-1}) \\
    &=  (\free_{t-1}-a)+ (\risky_{t-1} + a) \mathbb{E}_{P}\left(\frac{q(\omega_l)}{p(\omega_l)} | K_1, \dots, K_{t-1} \right) \\ 
    &= \free_{t-1} - a + \risky_{t-1} + a = \lvert \mathbf{K}_{t-1} \rvert.
\end{align*}

Similarly, let $b$ denote the value of shares of $K_t$ shorted at time $t$. If we limit $b \leq \frac{\free_t + u \risky_t}{u-1}$, then a similar analysis follows as above.

\end{proof}

\section{Proof of Proposition~\ref{prop:measure}}\label{appdx:proof-measure}
\begin{proof}
    Under the null hypothesis, the probability of of observing $P(X_t =1) = p$ and $P(X_t =0) = 1-p$ by definition. 

The multiplicative update is either  $1 + \lambda_t(1- p)$ or $1 + \lambda_t(0 - p)$. Assuming that the risk free asset has a rate of return $r = 0$, the risk-neutral probabilities $(q_t, 1-q_t)$ are the solution to 
$$ M_{t-1} = q_t (M_{t-1} (1 + \lambda_t(1- p))) + (1-q_t)(M_{t-1} (1 + \lambda_t(0 - p))).$$

Simplifying yields,
\begin{align*}
    1 &= q_t (1 + \lambda_t(1-p) ) + (1-q_t)(1 + \lambda_t(0 - p))\\
    &= q_t(1 - \lambda_t p) + (1-\lambda_t p) - q_t( 1- \lambda_t p) + q_t \lambda_t.
\end{align*}

Rearranging terms givens the result that $q_t = p$. Hence, for Bernoulli experiments, the risk-neutral probabilities are identical to the true probabilities under the null hypothesis.
\end{proof}

\section{Pricing a Call Option}\label{appdx:call-option}

We illustrate an example of the pricing of a speculative option contract in 
Figure~\ref{fig:put-price},
where we price a call option contract with expiry $\tau = 3$, strike price $S = 10/8$ and $K_0 = 1$. 
Since $S > K_0$, this option is more valuable for the speculator than for the hedger.
At expiry the value of the option is non-zero only if the outcomes of the Bernoulli experiment is $1$ at all three time periods.
The price of the option is only $17/64$, so for a price of less than $K_0 = 1$ the speculator can take a stake in the random sequence that pays $17/8$ after exercising the option and liquidating the shares.
The total profit from the transaction is $119/64$ which is a return of $600\%$.

\begin{figure}[htbp]
    \centering
    \includegraphics[width=0.6\textwidth]{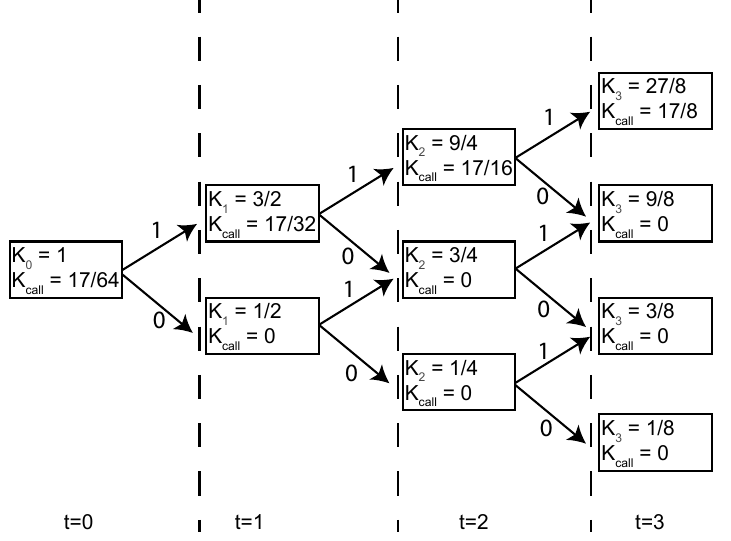}
    \caption{Calculating the price of a call option with expiry of $\tau = 3$, with a strike price of $S = 10/8$ at time $t=0$. The value of the call option at (i) follows from $K_{call} = \max\{0,K_3 - S\}= 27/8 - 10/8$. At (ii) the expected value of the contract is $0.5 (17/8) + 0.5 (0) = 17/16$. A similar process is repeated at (iii) and (iv) to have a price at time $t = 0$ of $17/64$. By only risking $17/64$, the investigator has the opportunity to purchase the risky asset a discounted price if favorable results occur.}
    \label{fig:put-price}
\end{figure}

\section{Proof of Lemma~\ref{lemma:options-portfolio}}\label{appdx:deriv-proof}
\begin{proof}
    By assumption $\lvert \mathbf{K}_0 \rvert = 1$, $\lvert \mathbf{K}_t \rvert$ is non-negative for each $t$ by design. What remains to be shown is that the conditional expectation is the previous value. 

    \begin{align*}
    \mathbb{E}_P(\lvert \mathbf{K}_t \rvert | \mid K_1, \dots, K_{t-1}) &= \mathbb{E}_P (\free_t + \risky_t + \deriv_t | K_1, \dots, K_{t-1}) \\
    &=  \free_{t-1} + \risky_{t-1}\mathbb{E}_P (1 + fX_l(\omega_l)|K_1, \dots, K_{t-1}) + \mathbb{E}_p (\deriv_{t} | K_1, \dots, K_{t-1}) \\
    &\stackrel{(i)}{=} \free_{t-1} + \risky_{t-1} + \deriv_{t-1} = \lvert \mathbf{K}_{t-1}\rvert
    \end{align*}
    The equality $(i)$ follows from the assumption the derivative contract is valued by its risk-neutral price and the risk-neutral measure is indistinguishable from the measure under the null hypothesis. 
\end{proof}

\section{Normal Test Martingale}\label{appdx:normal_martingale}

\begin{proof}
    $K_0 := 1$.
    Assume $\lambda$ is nonnegative. Then if $\lambda_t$ is fixed, then $(1 + \lambda_t(Y_t - e^{\frac{1}{2}}))$ is minimized by minimizing $Y_t$, which is bounded below by $0$. Then requiring $(1 + \lambda_t(0 - e^{\frac{1}{2}})) \geq 0$ corresponds to requiring $\lambda_t \leq e^{-\frac{1}{2}}$. By requiring $\lambda_t \in [0, e^{-\frac{1}{2}}]$, $K_t$ is a product of nonnegative numbers, and is nonnegative.

    \begin{align*}
        \mathbb{E}_{P}(K_T | K_1,\dots ,K_{T-1}) &= \mathbb{E}_{P}\left(\prod_{t=1}^{T}(1 + \lambda_t(Y_t - e^{-\frac{1}{2}})) \Bigm| K_1, \dots, K_{T-1} \right)\\
        &= K_{T-1}\mathbb{E}_{P}\left(1 + \lambda_T(Y_t - e^{-\frac{1}{2}})\Bigm| K_1, \dots, K_{T-1} \right) \\
        &= K_{T-1} (1 + \lambda_T(\mathbb{E}_{P}\left(Y_t \mid K_1, \dots, K_{T-1} \right)- e^{-\frac{1}{2}}))\\
        &= K_{T-1} (1 + \lambda_T( e^{-\frac{1}{2}}- e^{-\frac{1}{2}})) = K_{T-1}
    \end{align*}

\end{proof}

\end{document}